\documentclass[draftclsnofoot,onecolumn]{IEEEtran}
\IEEEoverridecommandlockouts
\usepackage{amsfonts}
\usepackage{dsfont}
\usepackage{amssymb}
\usepackage{graphicx}
\usepackage{subfigure}
\usepackage{enumerate}
\usepackage{amsmath}
\usepackage{color}
\usepackage{amsthm}
\usepackage{amsmath}
\usepackage{algorithm}
\usepackage{algpseudocode}
\usepackage[bookmarks=false]{hyperref}
\hypersetup{hidelinks}
\usepackage{breakurl}
\usepackage{cite}

\newcommand{\bp}{\begin{proof} \small }
\newcommand{\ep}{\end{proof} \normalsize}
\newcommand{\epx}{\end{proof} \small}
\newcommand{\bpa}{\begin{proofappx} \footnotesize }
\newcommand{\epa}{\end{proofappx} \small }
\newtheorem{theorem}{Theorem}

\newtheorem*{theorem*}{Theorem}
\newtheorem*{proposition*}{Proposition}
\newtheorem*{corollary*}{Corollary}
\newtheorem*{lemma*}{Lemma}
\newtheorem*{assumption*}{Assumption}
\newtheorem*{definition*}{Definition}
\newtheorem*{claim*}{Claim}

\newcommand{\be}{\begin{equation}}
\newcommand{\ee}{\end{equation}}
\newcommand{\bs}{\begin{subequations}}
\newcommand{\es}{\end{subequations}}
\newcommand{\bq}{\begin{eqnarray}}
\newcommand{\eq}{\end{eqnarray}}
\newcommand{\bqn}{\begin{eqnarray*}}
\newcommand{\eqn}{\end{eqnarray*}}

\newcommand{\ba}{\left[ \begin{array}}
\newcommand{\ea}{\\ \end{array} \right]}
\newcommand{\ben}{\begin{enumerate}}
\newcommand{\een}{\end{enumerate}}

\def\g{{\boldsymbol{g}}}

\def\w{{\boldsymbol{w}}}
\def\x{{\boldsymbol{x}}}
\def\y{{\boldsymbol{y}}}
\def\z{{\boldsymbol{z}}}

\def\real{{\mathchoice%
{\hbox{\rm\setbox1=\hbox{I}\copy1\kern-.45\wd1 R}}
{\hbox{\rm\setbox1=\hbox{I}\copy1\kern-.45\wd1 R}}
{\hbox{\scriptsize\rm\setbox1=\hbox{I}\copy1\kern-.45\wd1 R}}
{\hbox{\scriptsize\rm\setbox1=\hbox{I}\copy1\kern-.45\wd1 R}}}}

\def\Zint{{\mathchoice{\setbox1=\hbox{\sf Z}\copy1\kern-.75\wd1\box1}
{\setbox1=\hbox{\sf Z}\copy1\kern-.75\wd1\box1}
{\setbox1=\hbox{\scriptsize\sf Z}\copy1\kern-.75\wd1\box1}
{\setbox1=\hbox{\scriptsize\sf Z}\copy1\kern-.75\wd1\box1}}}
\newcommand{\complex}{ \hbox{\rm C\kern-0.45em\rule[.07em]{.02em}{.58em}%
\kern 0.43em}}

\begin{document}
	%
\title{Energy-Aware Analog Aggregation for \\Federated Learning with Redundant Data}
	
\author{\IEEEauthorblockN{Yuxuan Sun$^*$, Sheng Zhou$^*$, Deniz G\"und\"uz$^\dagger$}\\
\IEEEauthorblockA{$^*$Beijing National Research Center for Information Science and Technology\\
		Department of Electronic Engineering, Tsinghua University, Beijing 100084, China\\
		$^\dagger$Department of Electrical and Electronic Engineering, Imperial College London, London SW7 2BT, UK\\
		Email: \{sunyx15@mails., sheng.zhou@\}tsinghua.edu.cn, d.gunduz@imperial.ac.uk}
	}

\maketitle

\begin{abstract}
Federated learning (FL) enables workers to learn a model collaboratively by using their local data, with the help of a parameter server (PS) for global model aggregation. The high communication cost for periodic model updates and the non-independent and identically distributed (i.i.d.) data become major bottlenecks for FL. In this work, we consider analog aggregation to scale down the communication cost with respect to the number of workers, and introduce data redundancy to the system to deal with non-i.i.d. data. We propose an online energy-aware dynamic worker scheduling policy, which maximizes the average number of workers scheduled for gradient update at each iteration under a long-term energy constraint, and analyze its performance based on Lyapunov optimization. Experiments using MNIST dataset show that, for non-i.i.d. data, doubling data storage can improve the accuracy by $9.8\%$ under a stringent energy budget, while the proposed policy can achieve close-to-optimal accuracy without violating the energy constraint.
\end{abstract}


	
%
\IEEEpeerreviewmaketitle

\section{Introduction}

With the rapid development of machine learning (ML) techniques, emerging applications, including virtual and augmented reality, Internet of things, autonomous driving and e-health services, are penetrating into human lives \cite{park2018wireless}. ML models for these applications are typically trained in central clouds. However, centralized training leads to high communication costs, and causes privacy concerns in applications that involve sensitive personal data.

Meanwhile, the end devices such as smart phones, vehicles and sensors and the infrastructures like base stations (BSs) and road side units are being equipped with more computing resources, enabling intensive computations at the network edge, namely  \emph{multi-access edge computing} \cite{yao2017mec,li2018iot,zhou2019exploiting}.
With the help of edge intelligence and to address the privacy concerns, a distributed ML framework called \emph{federated learning} (FL) has been proposed recently \cite{googleai,flatscale,li2019federated}, where end devices, called \emph{workers}, learn a shared ML model collaboratively using their local data, with the help of a central \emph{parameter server (PS)} which aggregates the global model and coordinates the training process. Since the PS acquires the model update from each worker rather than their data, the privacy is preserved. 

The high communication costs and non-independent and identically distributed (i.i.d.) data are the two major bottlenecks in FL \cite{li2019federated}.
According to \cite{zhao2018federated}, when using highly non-i.i.d. data for FL, the accuracy drops by $11\%$ for MNIST and $51\%$ for CIFAR-10 as compared to using i.i.d. data. 
They prove that the non-i.i.d. level, i.e., the difference between the local and global data distributions, is the root cause of the performance degradation.
This problem is tackled by sharing publicly available i.i.d. data with the workers in \cite{zhao2018federated}, or workers sharing a limited portion of their data with the PS in \cite{yoshida2019hybrid}.

The communication burden of FL mainly comes from the global model aggregation, which can be reduced by efficient scheduling and resource allocation \cite{yang2019scheduling,zeng2019energy,wang2019adaptive,tran2019federated, abad2019hierar}, gradient quantization and sparsification \cite{abad2019hierar, mma2019machine,mma2019federated}, or via analog aggregation \cite{mma2019machine,mma2019federated,zhu2019broadband}.
An analytical study on the convergence rate achieved by random, round robin and proportional fair scheduling policies is carried out in \cite{yang2019scheduling}. An energy-efficient bandwidth allocation and worker scheduling scheme is proposed in \cite{zeng2019energy}, minimizing the energy consumption while maximizing the fraction of workers scheduled. A more general resource constraint, including both communication and computing resources, is considered in \cite{wang2019adaptive,tran2019federated}.
In \cite{abad2019hierar}, a hierarchical FL architecture is proposed, and the end-to-end latency is minimized by jointly considering model sparsification and the two-tier update interval. 
Quantization and error accumulation techniques are further considered in \cite{mma2019machine,mma2019federated} to reduce the communication cost.

Most papers on FL consider digital transmission for global model aggregation. However, the communication latency scales with the number of workers \cite{zhu2019broadband}. Observing that the PS is interested only in the \emph{average} of local models rather than their individual values, a promising solution is to use \emph{analog aggregation} which exploits the \emph{waveform-superposition property} of a wireless multiple access channel (MAC) \cite{mma2019machine,mma2019federated,zhu2019broadband}. If the workers synchronize with each other and align the transmit power, the summation of local models can be carried out \emph{over-the-air}. The tradeoff between signal-to-noise ratio (SNR) and the amount of exploited data is analyzed in \cite{zhu2019broadband}, while gradient compression and error accumulation are considered in \cite{mma2019machine,mma2019federated} to further improve the bandwidth efficiency of FL. While over-the-air aggregation requires channel state information at the workers, it is shown in \cite{mma2019collaborative} that this requirement can be released if the PS has multiple antennas.

Existing papers on analog aggregation mainly consider power allocation under specific channel models, and have not addressed the non-i.i.d. data.
In this work, we consider \emph{analog aggregation} for FL, where each worker has a \emph{long-term energy budget}, and \emph{data redundancy} is introduced to the system via data exchange or overlapped data collection. We propose an \emph{energy-aware dynamic} worker scheduling policy, which maximizes the average weighted fraction of scheduled workers without assuming specific channel models or requiring any future information, and analyze its performance based on Lyapunov optimization. Experiments using MNIST dataset show that, data redundancy can bring significant accuracy improvement when data is non-i.i.d., while the proposed policy can smartly utilize the available energy to achieve close-to-optimal accuracy.

The rest of the paper is organized as follows. In Sec. II, we introduce the system model and problem formulation. The worker scheduling policy is proposed in Sec. III, along with its performance analysis. Experiment results are presented in Sec. IV, and the paper is finally summarized in Sec. V.

\section{System Model and Problem Formulation}
\subsection{Federated Learning Architecture}

 \begin{figure*}[!t]
	\centering	
	\subfigure[Architecture of FL, where redundancy comes from data exchange between workers.]{\label{sys1}			
		\includegraphics[width=0.55\textwidth]{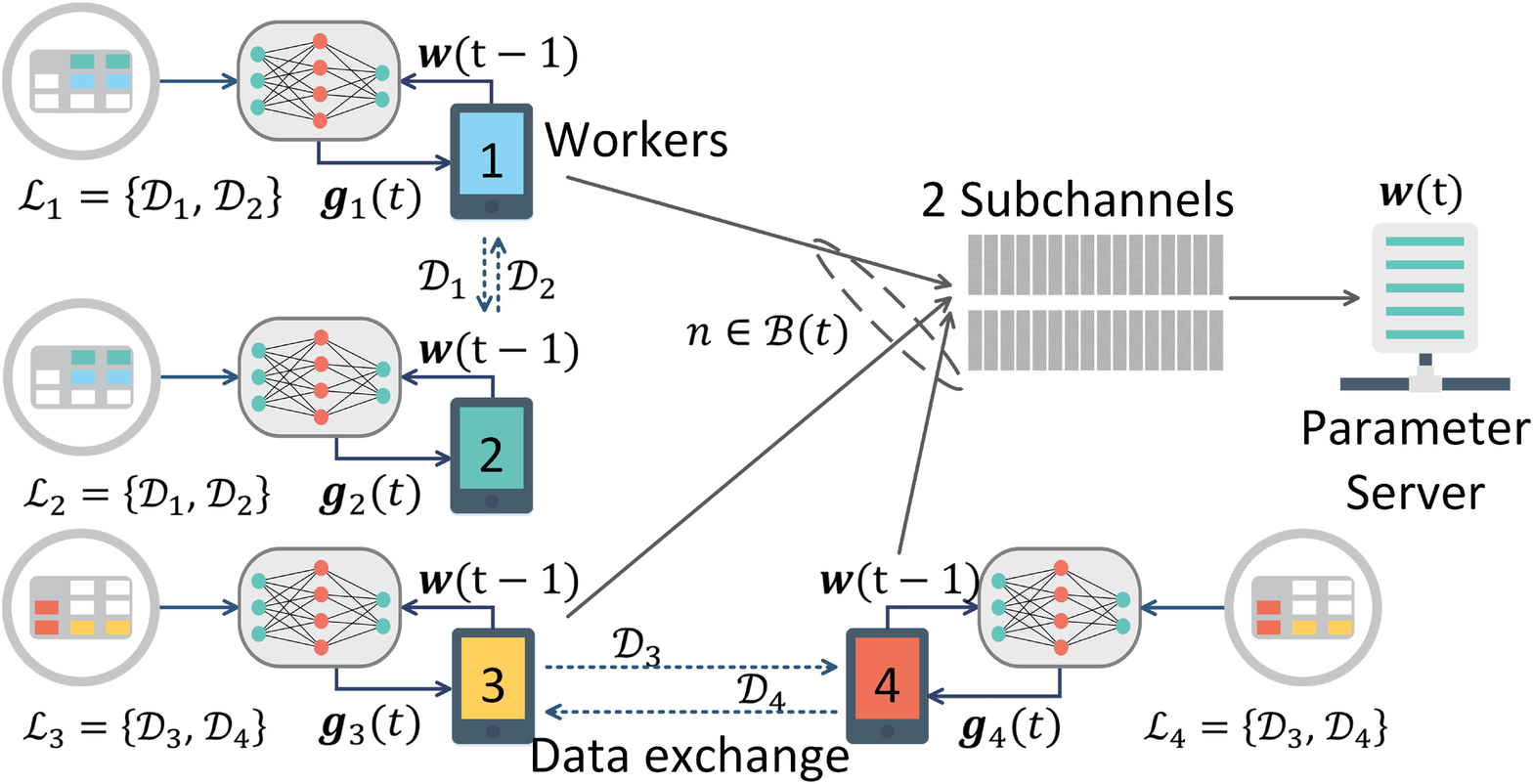}} 	
	\hspace{7mm}
	\subfigure[Overlapped data collection.]{\label{sys_re}	
		\includegraphics[width=0.28\textwidth]{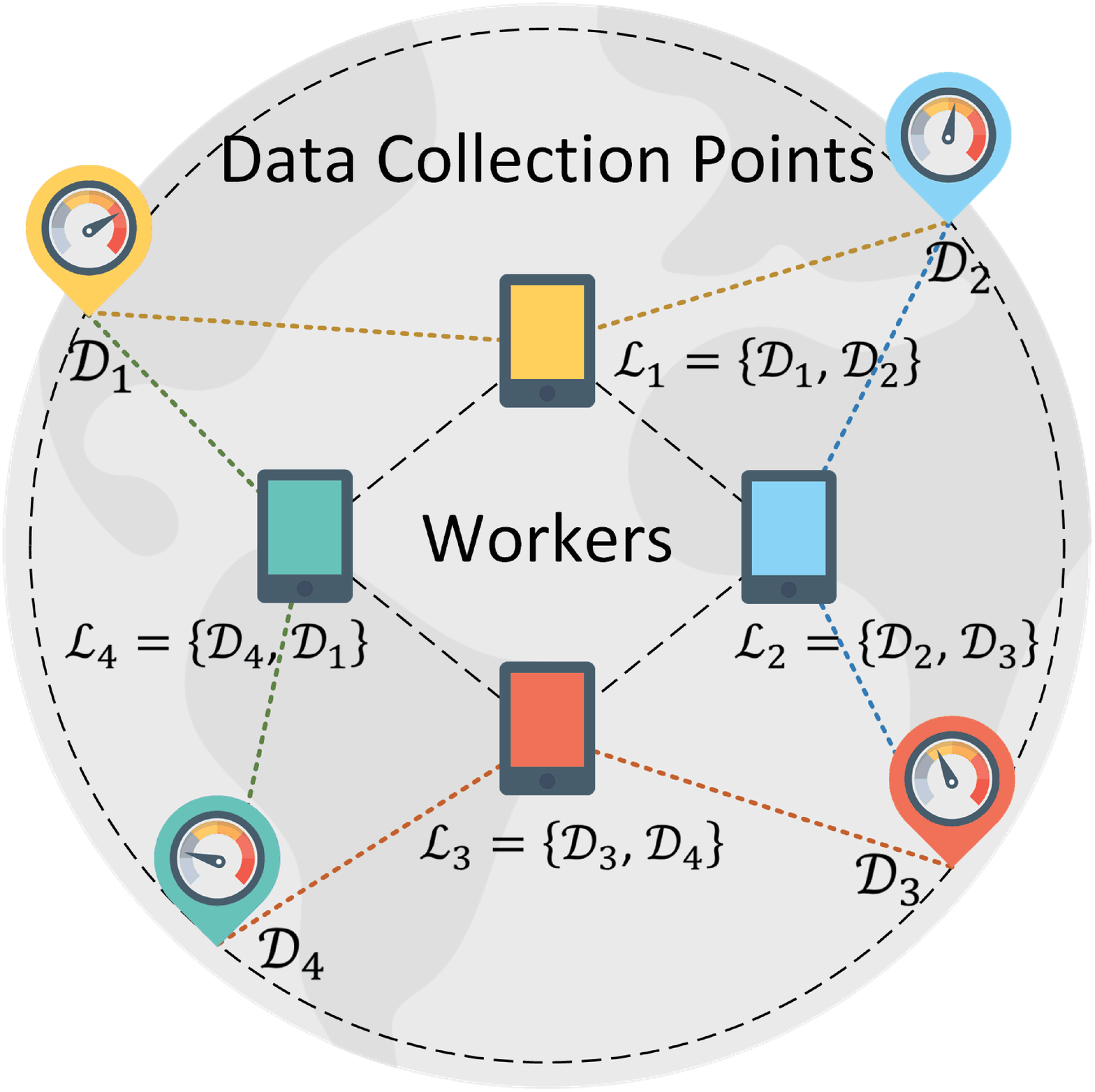}} 
	\caption{Illustration of a FL system and the acquisition of data redundancy.}
	\label{system}
\end{figure*}

As shown in Fig. \ref{sys1}, we consider an FL system with a single PS and $N$ homogeneous workers $\mathcal{N}=\{1,2,\cdots,N\}$. 
To tackle non-i.i.d. data, we consider introducing \emph{data redundancy} to the system via 1) \emph{data exchange}, i.e., workers exchange their local data with neighboring workers they trust, or 2) \emph{overlapped collection}, i.e., in IoT networks, the sensing area of sensors (workers) are overlapped with each other. For example, in Fig. \ref{sys_re}, there are 4 data collection points, and each worker collects data from 2 points.
As opposed to \cite{zhao2018federated} and \cite{yoshida2019hybrid}, this does not require sharing any data samples at the PS.
We assume that there are $K$ original datasets $\mathcal{D}_{1}, \cdots, \mathcal{D}_{K}$ generated by workers (in data exchange case) or data collection points (in overlapped collection case), each has the same number of data samples, denoted by $D$. 
The global dataset is defined as $\mathcal{D}=\bigcup_{k=1,2,\cdots, K}\mathcal{D}_k$, with $K\!D$ data samples $\{\x_1,\cdots, \x_{K\!D}\}$.
For simplicity and without loss of generality, we assume $K=N$ in the following.

The data redundancy of the system is denoted by $r$, indicating that each original dataset is stored at $r$ different workers.
The local dataset owned by worker $n$ is denoted by $\mathcal{L}_{n}$, with $|\mathcal{L}_{n}|=rD$.
One example to obtain redundancy $r$ is to collect or exchange the original datasets in a cyclic manner: Let worker $n$ stores $\mathcal{L}_{n}=\bigcup_{n\in \mathcal{H}_n}\mathcal{D}_n$, where the set of indexes is $\mathcal{H}_n=\{H(n),~\cdots, H(n+r-1)\}$, with
\begin{align} \label{cyclic_r}
	H(n)=
	\begin{cases}
	n, &1\leq n \leq N,  \\
	n-N, &n>N.
	\end{cases}
\end{align}
Fig. \ref{sys_re} is an example of obtaining redundancy $r=2$ in the data collection scenario, with $K=N=4$.

The goal of FL is to minimize the global loss 
\begin{align}
\min_{\w} F(\w)\triangleq\frac{1}{ND} \sum_{i\in\mathcal{D}} f(\w,\x_i),
\end{align}
where $f(\w,\x_i)$ is a loss function designed for the FL task, and $\w\in \mathbb{R}^{s}$ is the parameter vector to be optimized.

In the $t$-th training round, the PS broadcasts the global parameter vector $\w(t-1)$  obtained in the last round to all the workers. 
We assume that the PS is a more capable node with sufficient energy resources (e.g., a BS); therefore the broadcast of the global parameter vector is error-free.
Each worker randomly picks up a fraction $\lambda_n(t)$ of data samples $\mathcal{L}_n(t)\subseteq \mathcal{L}_n$, with $|\mathcal{L}_n(t)|=\lambda_n(t) rD$, to evaluate its local gradient estimate $\g_n(t)$ as:
\begin{align} \label{up_gra}
\g_n(t)=\frac{1}{\lambda_n(t) rD} \sum_{\x_i\in\mathcal{L}_n(t)}\nabla f\left(\w(t-1),\x_i\right) .
\end{align}
Here we can let $\lambda_n(t)=\frac{1}{r}$, so that data redundancy \emph{does not} bring additional computing workloads to the workers for training, but only increases the storage cost.

We define an indicator function $\beta_n(t)$, where $\beta_n(t)=1$ if worker $n$ is scheduled to upload gradient in the $t$-th round, and $\beta_n(t)=0$ otherwise. 
Further define the set of workers scheduled in round $t$ as $\mathcal{B}(t)=\{n\in\mathcal{N}: \beta_n(t)=1\}$. The global parameter vector $\w(t)$ is updated according to
\begin{align}\label{up_model}
\w(t)=\w(t-1)-\eta_t \frac{1}{|\mathcal{B}(t)|}\sum_{n\in\mathcal{B}(t)} \boldsymbol{g}_n(t),
\end{align}
where $\eta_t$ is the learning rate.

\subsection{Analog Aggregation}
For the aggregation of the local gradients, we consider \emph{analog transmission} via a wireless MAC with $M$ sub-channels.
If worker $n$ is scheduled, its local gradient $\g_n(t)$ is evenly partitioned into $M$ segments $\g_n(t)=\left[\g_{1,n}(t),~\cdots,~\g_{M,n}(t)\right]$, where, for $m=1,\cdots,M$, $\g_{m,n}(t)$ is a vector with either $\left\lceil \frac{s}{M}\right\rceil$ or $\left\lfloor \frac{s}{M}\right\rfloor$ entries, and transmitted via sub-channel $m$.

In order to carry out the summation of the local gradients over-the-air, all the scheduled workers need to be synchronized and align their transmit power. Specifically, in round $t$, denote the power allocated to worker $n$ within sub-channel $m$ by $p_{m,n}(t)$, which satisfies
\begin{align}
p_{m,n}(t)=\frac{\beta_n(t)\sigma_m(t)}{h_{m,n}(t)},
\end{align}
where $h_{m,n}(t)$ is the channel gain between worker $n$ and the PS in sub-channel $m$, and $\sigma_m(t)$ is a power scalar to determine the received SNR.
We assume that $h_{m,n}(t)$ remains constant within each round, but we do not limit $h_{m,n}(t)$ to any specific distribution. 
We also assume that each worker $n$ has perfect knowledge of its current channel gains $h_{m,n}(t)$, $\forall m$.
Within sub-channel $m$, each scheduled worker $n \in \mathcal{B}(t)$ transmits $p_{m,n}(t)\g_{m,n}(t)$ to the PS.
The total communication latency in each round is $\left\lceil \frac{s}{M}\right\rceil$ times the symbol duration, regardless of the number of workers scheduled. 
We remark that, the consideration of sub-channels enables us to implement analog aggregation in the current digital transmit systems such as orthogonal frequency division multiplexing (OFDM) with minor changes \cite{zhu2019broadband}. 
However, we consider a worker-level schedule $\beta_n(t)$ rather than a sub-channel-worker-level schedule $\beta_{m,n}(t)$ in this work, so that the PS can receive the whole gradients of the scheduled workers, as shown in \eqref{up_model}.

At the PS side, the received signal over sub-channel $m$ can be written as
\begin{align} \label{agg_y}
\y_m(t)=\sum_{n\in\mathcal{B}(t) } h_{m,n}(t) p_{m,n}(t) \g_{m,n}(t)+\z_m(t) =\sigma_m(t)\sum_{n\in\mathcal{B}(t)}\g_{m,n}(t)+\z_m(t),
\end{align}
where $\z_m(t)$ is an i.i.d. additive white Gaussian noise (AWGN) vector, with each entry following the standard normal distribution.
The $m$-th segment of the global parameter vector, $\w_m(t)$, is updated according to
\begin{align} \label{model_update}
\w_m(t)&=\w_m(t-1)-\frac{\eta_t}{\left|\mathcal{B}(t)\right|\sigma_m(t)}\y_m(t) \nonumber\\
&=\w_m(t-1)-\frac{\eta_t}{\left|\mathcal{B}(t)\right|} \sum_{n\in\mathcal{B}(t)}\g_{m,n}(t) - \frac{\eta_t}{\left|\mathcal{B}(t)\right|\sigma_m(t)}\z_m(t).
\end{align}
And finally, the global parameters are concatenated into vector $\w(t)=\left[\w_1(t),~\cdots,~\w_M(t)\right]$.


\subsection{Problem Formulation}
Our objective is to minimize the global loss $F(\w(T-1))$ after $T$ training rounds, by optimizing the worker schedule $\{\beta_n(t)\}$ and power allocation $\{\sigma_m(t),p_{m,n}(t)\}$. Meanwhile, we also want to explore how data redundancy $r$ impacts the performance. The problem is formulated as:
\begin{subequations}
\begin{align}
&\mathcal{P}1: \min_{\left\{\beta_n(t),~\sigma_m(t),~p_{m,n}(t)\right\}} F(\w(T-1)) \\
&\text{s.t.} ~~~~ \frac{1}{T}\sum_{t=0}^{T-1}\sum_{m=1}^{M}\left\lVert   p_{m,n}(t) \g_{m,n}(t)\right\rVert_2^2 \leq \bar{E}_n,  ~ \forall n, \label{cons_en} \\
&~~~~~~~~p_{m,n}(t)h_{m,n}(t)=\sigma_m(t)\beta_n(t),~\forall m,n, t, \label{cons_align}\\
&~~~~~~~~ \beta_n(t)\in\{0,1\}, \forall n, t. \label{cons_beta}
\end{align}
\end{subequations}
The first constraint \eqref{cons_en} states that the average energy consumed by each worker in each training round cannot exceed budget $\bar{E}_n$, due to the battery limitation of wireless devices.\footnote{We unify each round to a unit time length, and use power and energy interchangeably in this paper without any ambiguity.}
The second constraint \eqref{cons_align} states that all the scheduled workers align their power to enable over-the-air computation. 

Since the loss function $F(\w)$ is usually different for different kinds of machine learning tasks, and the evolution of the parameters during the training process is very complex, it is hard to express $F(\w(T-1))$ explicitly.
Meanwhile, the convergence rate of distributed SGD is found to be positively correlated to the number of workers scheduled, as shown in \cite{zeng2019energy} and references therein. Therefore, we consider an alternative optimization problem that maximizes the average weighted fraction of scheduled workers:
\begin{align}
	\max \frac{1}{T}\sum_{t=0}^{T-1}\frac{\gamma(t)\sum_{n=1}^{N}\beta_n(t)}{N}, 
\end{align}
where $\gamma(t)$ characterizes the importance of scheduling more workers in the $t$-th round.

We further fix $\sigma_m(t)$ to a predefined value $\sigma$ for $\forall t$. 
This assumption is based on the fact that, with analog aggregation, the convergence speed of the FL task is not very sensitive to the SNR or the average transmit power, according to \cite{mma2019federated}. 
Then $p_{m,n}(t)=\frac{\sigma\beta_n(t)}{h_{m,n}(t)}$ can be obtained from \eqref{cons_align} after deciding whether to schedule worker $n$ or not. Finally, the energy consumption is given by
\begin{align} \label{energy}
	E_n(t)=\sum_{m=1}^{M}\left\lVert   \frac{\sigma}{h_{m,n}(t)} \g_{m,n}(t)\right\rVert_2^2.
\end{align}
The alternative problem can be formulated as
\begin{subequations}
	\begin{align}
	\mathcal{P}2: &\min_{\left\{\beta_n(t)\right\}}~ \frac{1}{T}\sum_{t=0}^{T-1}u(t)\triangleq\gamma(t)\left( 1- \frac{\sum_{n=1}^{N}\beta_n(t)}{N} \right) \\
	&~~~\text{s.t.} ~~~\frac{1}{T}\sum_{t=0}^{T-1}\beta_n(t) E_n(t)\leq \bar{E}_n, ~\forall n, \label{cons_en_p2}\\
	&~~~~~~~~~~\beta_n(t)\in\{0,1\},~\forall n,t.
	\end{align}
\end{subequations}

\section{Energy-aware Worker Scheduling}
The key challenge to solve $\mathcal{P}2$ is that, constraint \eqref{cons_en_p2} is a long-term energy budget. However, in practice, the channel gains $h_{m,n}(t)$ and the power of the gradients $\left\lVert \g_{m,n}(t)\right\rVert_2^2$ cannot be acquired before the $t$-th round, and they may not be i.i.d. over time. Therefore, we design \emph{online} worker scheduling policies in this section, and carry out performance analysis without assuming any specific distributions for the channel.

For any worker $n$, it is easy to see that, its energy constraint \eqref{cons_en_p2} and the scheduling decision $\beta_n(t)$ are independent of other workers. Then $\mathcal{P}2$ can be \emph{equivalently} decoupled into $N$ individual problems
\begin{subequations}
	\begin{align}
	\mathcal{P}3: &\min_{\left\{\beta_n(t)\right\}} \frac{1}{T}\sum_{t=0}^{T-1}u_n(t)\triangleq \gamma(t) \frac{ 1- \beta_n(t)}{N} \\
	&~~~\text{s.t.}~~ \frac{1}{T}\sum_{t=0}^{T-1}\beta_n(t) E_n(t) \leq \bar{E}_n, \label{cons_en_p3} \\
	&~~~~~~~~~\beta_n(t)\in\{0,1\},~\forall t.
	\end{align}
\end{subequations}
The combination of the optimal solution of  $\mathcal{P}3$ for all workers is the optimal solution of $\mathcal{P}2$. And by solving $\mathcal{P}3$, each worker can decide whether or not to update gradient \emph{individually}.
In what follows, we design online solutions to $\mathcal{P}3$.

\subsection{Myopic Scheduling for Short-term Fixed Energy Constraint}


\begin{algorithm} [!t]
	\caption{Energy-Aware Dynamic Scheduling Policy for FL via Analog Aggregation}
	\begin{algorithmic}[1]
		\State \textbf{Initialization}: initialize global model $\w(-1)$, input $N$, $\sigma$, $ \bar{E}_n$, $\gamma(t)$, $V$,  and let $q_n(0)=q_{\text{min}}$, $\forall n$.
		\For {$t=0,1,\cdots, T-1$}
			\State Broadcast $\w(t-1)$ to all the workers. \Comment{\textit{PS}}
			\State Update $\g_n(t)$ from \eqref{up_gra}. \Comment{\textit{Each worker, in parallel}}
			\State Acquire channel gains $\{h_{m,n}\}$ for $m=1,\cdots, M$, and calculate energy consumption $E_n(t)$ according to \eqref{energy}.
			\State Make scheduling decision:
			\begin{align}\label{beta_t}
					\beta_n(t)=
				\begin{cases}
				1,~&\text{if }q_n(t)E_n(t) \leq \frac{V\gamma(t)}{N},\\
				0, ~ &\text{if }q_n(t)E_n(t) > \frac{V\gamma(t)}{N}.  
				\end{cases}
			\end{align}
			\State Update virtual queue $q_n(t)$ according to \eqref{queue_evo}.
			\State Transmit $\frac{\sigma}{h_{m,n}(t)}\g_{m,n}(t)$ in sub-channel $m$, $\forall m$.
			\State Aggregate received signal $\y_m(t)$ according to \eqref{agg_y}, and update global model $\w(t)$ according to \eqref{model_update}.\Comment{\textit{PS}}
		\EndFor
	\end{algorithmic}
\end{algorithm}

A simple way to handle the average energy constraint \eqref{cons_en_p3} is to remove the long-term summation, and transform it to a short-term fixed energy constraint $\beta_n(t) E_n(t) \leq \bar{E}_n$, for $\forall t$.
Then the \emph{myopic scheduling} policy can be given by
\begin{align}
\beta_n(t)=
\begin{cases}
1,~&E_n(t) \leq \bar{E}_n,\\
0, ~ &E_n(t) > \bar{E}_n.  
\end{cases}
\end{align}
In the $t$-th round, worker $n$ acquires the current channel gains $\{h_{m,n}(t)\}$ and the powers of the gradient $\{\lVert   \g_{m,n}(t)\rVert_2^2\}$ for $m=1,2,\cdots, M$, and calculates the required energy $E_n(t)$ to send its gradient estimate. If the required energy is no more than the budget $\bar{E}_n$, worker $n$ is scheduled.

\subsection{Energy-Aware Dynamic Scheduling for Long-term Average Energy Constraint}
Although the myopic scheduling policy is simple and can satisfy the original long-term average energy budget, it actually introduces a much tighter energy constraint. In other words, the worker uses its energy in a more conservative manner, and is less likely to be scheduled compared with that allowed by the original energy budget \eqref{cons_en_p3}.

To schedule workers more efficiently, we propose an \emph{energy-aware dynamic scheduling} policy, as shown in Algorithm 1. We construct a virtual queue $q_n(t)$ with $q_n(0)=q_{\text{min}}\geq 0$, and its evolution is given by
\begin{align} \label{queue_evo}
q_n(t+1)=\max\{q_n(t)+\beta_n(t)E_n(t)-\bar{E}_n, q_{\text{min}}\}.
\end{align}
The value of the virtual queue indicates the deficit between the current energy consumption and the budget.

As shown in Lines 3-4 in Algorithm 1, in each round, first, the PS broadcasts the up-to-date global parameter vector, and each worker runs a local gradient estimation step \emph{in parallel} based on its local data. In Lines 5-6, each worker makes the scheduling decision $\beta_n(t)$ by comparing the \emph{weighted energy} $q_n(t)E_n(t)$ and the \emph{weighted utility} $\frac{V\gamma(t)}{N}$, where $V$ is an adjustable weight parameter.
The virtual queue $q_n(t)$ transforms the long-term energy budget into instantaneous energy constraint: if $q_n(t)$ is large, it is more likely that $q_n(t)E_n(t)>\frac{V\gamma(t)}{N}$, so that the worker tends to not update the gradient to save energy; and vice versa.
Also, by introducing $q_n(t)$, $\beta_n(t)$ can be obtained without any future information, or other workers'  states and decisions. Therefore, the proposed algorithm is \emph{energy-aware}, \emph{online} and \emph{distributed}.
As shown in Lines 7-9, all the workers then update their virtual queues, and the scheduled workers transmit their gradients to the PS synchronously.
The global parameter vector is finally aggregated by the PS.

\subsection{Performance Analysis}
To analyze the performance of the proposed algorithm, we assume that $q_{\text{min}}=0$ in this subsection, and refer to the Lypunov optimization technique \cite{neely2010stochastic}. We consider a non-ergodic version of Lyapunov optimization, i.e., all the random variables can be non-i.i.d. across time; that is, 1) the distribution of the power of the gradient $||\g_{m,n}(t)||^2_2$ is unknown; 2) the channel and user mobility are not limited to specific models; 3) the total number of rounds $T$ is finite.

Define $u_n^*$ as the optimal utility of $\mathcal{P}3$ achieved by the offline genie-aided solution, and $u_n^{\dagger}$ the average utility achieved by Algorithm 1. Let $u^*=\sum_{n=1}^{N}u_n^*$ be the optimal utility of $\mathcal{P}2$,  and $u^\dagger=\sum_{n=1}^{N}u_n^\dagger$.
By applying the energy-aware dynamic scheduling policy, we get the following theorem.
\begin{theorem}
	When $q_{\text{min}}=0$, the average weighted fraction of scheduled workers achieved by Algorithm 1 satisfies:
	\begin{align} \label{bound_u}
		u^{\dagger}&\leq u^*+ \frac{T}{2V}\sum_{n=1}^{N} \alpha_n^2,
	\end{align}
	and the total energy consumption of worker $n$ is bounded by:
	\begin{align} \label{bound_en}
		\sum_{t=0}^{T-1} \beta_n(t)E_n(t) \leq T\bar{E}_n+\sqrt{T^2\alpha_n^2+2VTu_n^*},
	\end{align}
	where $\alpha_n=\max_t \left\{\left|\beta_n(t)E_n(t)-\bar{E}_n\right|\right\}$.
\end{theorem}
\begin{proof}
	See Appendix A.
\end{proof}

Theorem 1 shows that, the average utility and the total energy consumption achieved by the proposed energy-aware dynamic scheduling policy have deviation bounds, compared with the optimal genie-aided policy and the energy budget, respectively. Both deviations are positively correlated to the maximum energy deficit $\alpha_n$, and can be traded-off by the weight parameter $V$.

\subsection{Discussions}
We remark that, according to \eqref{beta_t}, when $q_{\text{min}}=0$, the scheduling indicator $\beta_n(t)=1$ if $q_n(t)=q_{\text{min}}$. However, the energy cost $E_n(t)$ may be very high, leading to a large upper bound on the energy deficit $\alpha_n$, and thus a large deviation from the optimal genie-aided solution. Moreover, the worker cannot be scheduled for many rounds afterwards, which in turn reduces the average utility. Therefore, although $q_{\text{min}}=0$ in the classical Lyapunov optimization, and the worst-case utility and energy consumption can be guaranteed, in the experiments, we set $q_{\text{min}}>0$, which is more reasonable. 

In Algorithm 1, the weight parameter $\gamma(t)$ is also involved in balancing the utility and energy consumption.
In practice, we can set $\gamma(t)$ as a relatively large value at the beginning of the training process, and decrease it across time, since 
1) Scheduling more workers at the initial training rounds helps the learning process converge faster;
2) The power of the gradients reduces across time, according to the experiments in the following;
3) The battery level decreases as time goes by, so that we should use energy more and more conservatively.

\section{Experiments}
In this section, we evaluate the gain of data redundancy $r$ and the performance of the proposed worker scheduling policy for a digit recognition FL task, using the MNIST dataset\footnote{http://yann.lecun.com/exdb/mnist/} with $60000$ training samples and $10000$ test samples. We consider $N=50$ workers, $M=100$ sub-channels and $T=100$ training rounds.
We divide the dataset in either an i.i.d. or a non-i.i.d. fashion. For i.i.d. data partition, the data samples are randomly partitioned into $N$ datasets $\mathcal{D}_1, \cdots, \mathcal{D}_N$. For the non-i.i.d case, the data samples are first sorted by the order of digits, and each dataset is formed by $1200$ data samples from a single digit. Each worker stores $r$ of these datasets in a cyclic manner according to \eqref{cyclic_r}, a total of $1200r$ samples. During each round, a fraction $\lambda_n(t)=\frac{1}{r}$ of data samples are randomly chosen from the $1200r$ data samples worker $n$ stores for training, so that the amount of computation for training is the same under different data redundancies.

\begin{figure}[!t]
	\centering
	\includegraphics[width=0.45\textwidth]{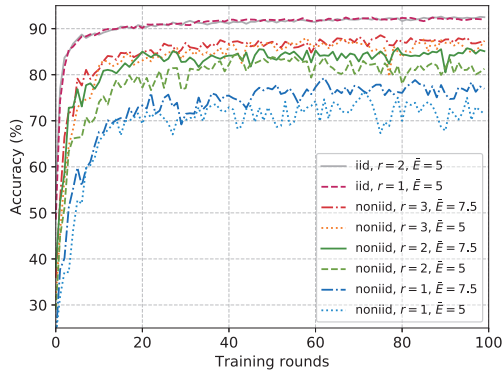}
		\vspace{-3mm}
	\caption{The accuracy of the MLP under different data redundancies using the myopic scheduling policy.}	\label{acc_fix}
\end{figure}

\begin{figure}[!t]
	\centering
	\includegraphics[width=0.45\textwidth]{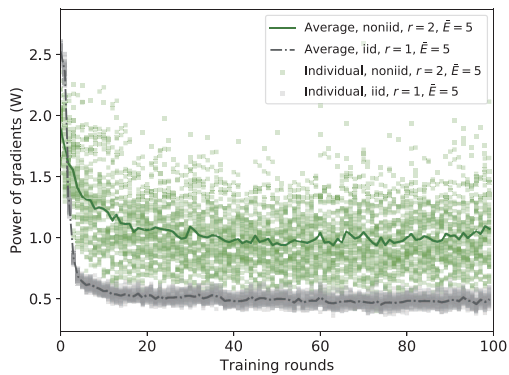}
	\vspace{-3mm}
	\caption{Variation of the power of the gradients with training rounds.}	\label{power}
		\vspace{-7mm}
\end{figure}

We train a multilayer perceptron (MLP) with a 784-neuron input layer, a 64-neuron hidden layer, and a 10-neuron softmax output layer. The total number of parameters is 50890. Cross entropy is adopted as the loss function, and rectified linear unit (ReLU) activation is used. We set the learning rate $\eta_t$ as 0.05, the dropout probability as 0.5, and use a momentum of 0.5. 
We consider Rayleigh fading channels from the workers to the PS, where the channel gain $h_{m,n}(t)$ follows the standard complex normal distribution, and each entry of the additive noise vector $\z_m(t)$ follows standard normal distribution. Also, the power scalar is set as $\sigma=1$.
For the energy-aware dynamic scheduling policy, we let $V=1500$, and $q_{\text{min}}=0.3$.

In Fig. \ref{acc_fix}, we first explore the impact of data redundancy by evaluating the accuracy of the MLP using the myopic scheduling policy under both i.i.d. and non-i.i.d. data. The energy budget of each worker is set the same, with $\bar{E}_n=\bar{E}=7.5~\mathrm{J}$ or $5~\mathrm{J}$. Overall, the accuracy with i.i.d. data outperforms that with non-i.i.d. data. When data is i.i.d., redundancy hardly brings any benefit to the system, as the workers can already use images of different digits to train the model even if $r=1$. However, data redundancy can significantly improve the accuracy of the MLP under non-i.i.d. data. Specifically:
1) Given the energy budget, as the data redundancy increases, the accuracy of the model increases with diminishing marginal gain. When $\bar{E}=5$, an improvement of $9.8\%$ is achieved by increasing redundancy from $r=1$ to $r=2$, and $6\%$ from $r=2$ to $r=3$.
2)  Increasing the energy budget helps to improve the performance of the MLP since more workers can be scheduled in each round, while the system is less sensitive to the energy constraint in the case of i.i.d. data. Comparing $\bar{E}=5$ with $\bar{E}=7.5$, the accuracy improves by $5.5\%$ when $r=1$, by $3.8\%$ when $r=2$, but less than $1\%$ when $r=3$.

In Fig. \ref{power}, we plot the power of the gradients $||\g_n(t)||_2^2$ to guide the parameter design of the energy-aware dynamic scheduling policy.
Each scatter point represents the power of the gradient of a worker in that round, and the lines are obtained by averaging these across workers, i.e., $\frac{1}{N}\sum_{n=1}^{N}||\g_n(t)||_2^2$. 
We find that, the power of the gradients reduces dramatically in the first 10-15 training rounds, and is quite stationary afterwards.
Motivated by this observation, the weight parameter $\gamma(t)$ is set as:
\begin{align}
	\gamma(t)=
	\begin{cases}
	2, &0\leq t < 9,  \\
	2-0.2(t-9), &10\leq t < 15, \\
	1,&t\geq 15.
	\end{cases}
\end{align}

 \begin{figure}[!t]
	\centering	
	\subfigure[Accuracy.]{\label{dyn_acc}	
		\includegraphics[width=0.45\textwidth]{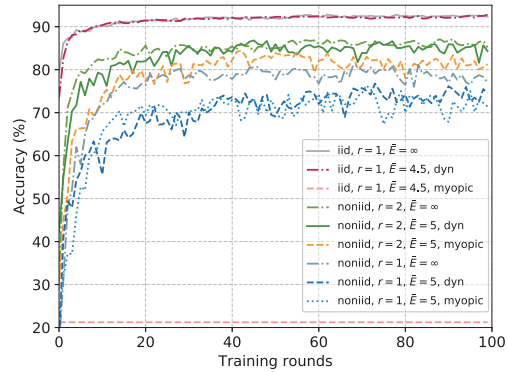}} \\
	\subfigure[Fraction of workers scheduled.]{\label{dyn_frac}	
		\includegraphics[width=0.45\textwidth]{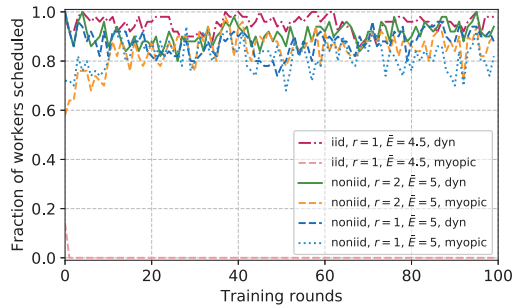}}\\ 
		\subfigure[Cumulative energy consumption.]{\label{dyn_en}	
		\includegraphics[width=0.45\textwidth]{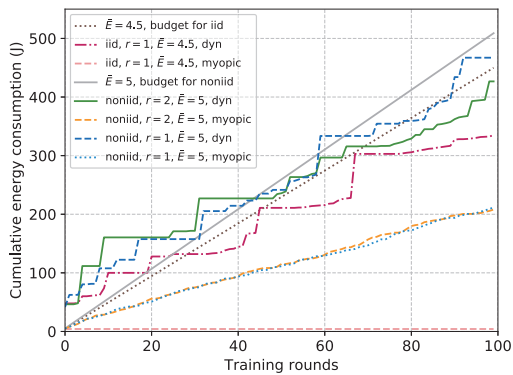}} 
	\caption{The performance of the proposed energy-aware dynamic scheduling policy, compared with the myopic scheduling policy and the upperbound.}
	\label{dyn}
	\vspace{-5mm}
\end{figure}

The performance of the proposed energy-aware dynamic scheduling policy is then presented in Fig. \ref{dyn}.
As shown in Fig. \ref{dyn_acc}, we compare the accuracy of the dynamic scheduling policy with the myopic policy and an upperbound, which is achieved by setting $\bar{E}=\infty$, so that all the workers can be scheduled in each round. We can see that, under both i.i.d. and non-i.i.d. data distributions, the proposed algorithm can achieve close-to-optimal accuracy, while outperforming the myopic policy. Note that under i.i.d. data distribution, we further reduce the energy budget to $\bar{E}=4.5~\mathrm{J}$.
Since the power of the gradient is high at the beginning of training, no worker can satisfy the energy constraint using the myopic policy. However, the dynamic scheduling policy enables workers to borrow energy from the future, so that workers can be scheduled.
Fig. \ref{dyn_frac} and Fig. \ref{dyn_en} present the fraction of workers  $\frac{\sum_{n=1}^{N}\beta_n(t)}{N}$ scheduled in each round, and the maximum cumulative energy consumption $\max_{\{n\in\mathcal{N}\}}\sum_{\tau=0}^{t}\beta_n(\tau)E_n(\tau)$ over workers, respectively. By using the dynamic scheduling policy, more workers can be scheduled compared with the myopic policy, and the energy can be fully utilized.
Specifically, with non-i.i.d. data, when $\bar{E}=5~\mathrm{J}$ and data redundancy is $r=2$, dynamic policy schedules an average of $90.9\%$ workers in each round, while myopic policy can only schedule $84.6\%$ workers.
Moreover, the dynamic scheduling policy schedules more workers at the beginning while the energy is sufficient, so that the convergence of the MLP is accelerated.

\section{Conclusions}
We have considered analog aggregation for FL over wireless channels and introduced data redundancy to the system to deal with non-i.i.d. data. We have proposed an energy-aware worker scheduling policy to maximize the weighted fraction of scheduled workers, which works in an online, distributed manner with performance guarantee. Experiments on the MNIST dataset have been carried out, showing that for non-i.i.d. data, increasing data redundancy to $r=2$ can improve the accuracy by $9.8\%$ under a stringent energy budget. Further increases in redundancy lead to diminishing improvements in accuracy. We have also shown that, the proposed energy-aware dynamic scheduling policy can achieve close-to-optimal performance without violating the energy budget, and on average schedule $6\%$ more workers than a heuristic myopic policy.
In the future, we plan to quantify the impact of data redundancy, and consider gradient compression and error accumulation to further reduce the communication cost.

\section*{Acknowledgment}
This work is sponsored in part by the European Research Council (ERC) under Starting Grant BEACON (grant No. 677854), the Nature Science Foundation of China (No. 61871254, No. 91638204, No. 61571265, No. 61861136003, No. 61621091), National Key R\&D Program of China 2018YFB0105005, and Intel Collaborative Research Institute for Intelligent and Automated Connected Vehicles. 

\appendices{}
\section{Proof of Theorem 1}
Let $y_n(t)=\beta_n(t)E_n(t)-\bar{E}_n$. From \eqref{queue_evo}, we have $y_n(t)\leq q_n(t+1)-q_n(t)$, $q_n^2(t+1) \leq \left(q_n(t)+y_n(t)\right)^2$, $|q_n(t+1)-q_n(t)|\leq |y_n(t)|$, and $(q_n(t+1)-q_n(t))y_n(t)\leq y_n^2(t)$.
When $q_{\text{min}}=0$, $q_n(0)=0$, and thus
\begin{align} \label{yn_sum}
\sum_{t=0}^{T-1}y_n(t)=\sum_{t=0}^{T-1}\beta_n(t)E_n(t)-T\bar{E}_n\leq q_n(T).
\end{align}
Define the Lyapunov function as $L_n(t)\triangleq \frac{1}{2} q_n^2(t)$, and the one-slot Lyapunov drift as:
\begin{align}
\Delta_n^{[1]}(t)&\triangleq L_n(t+1)-L_n(t)=\frac{1}{2}q_n^2(t+1)- \frac{1}{2}q_n^2(t) \nonumber\\
&\leq \frac{1}{2}y_n^2(t) + q_n(t)y_n(t) \leq    \frac{1}{2} \alpha_n^2+ q_n(t)y_n(t),
\end{align}
where $\alpha_n=\max_t \left\{ |y_n(t)|\right\}$. The one-slot drift-plus-penalty function is given by:
\begin{align} \label{drift_plus_penalty}
\Delta_n^{[1]}(t)+Vu_n(t)\leq \frac{1}{2} \alpha_n^2+ q_n(t)y_n(t)+Vu_n(t).
\end{align}

The main idea of the proposed dynamic scheduling policy is to minimize the right-hand-side of \eqref{drift_plus_penalty}.
We have:
\begin{align} \label{dyn_opt}  
\min_{\beta_n(t)}  q_n(t)y_n(t)+Vu_n(t)
\Leftrightarrow \min_{\beta_n(t)} \beta_n(t) \left(q_n(t)E_n(t) - \frac{V\gamma(t)}{N} \right)-q_n(t)\bar{E}_n+\frac{V\gamma(t)}{N}. 
\end{align}
Since $\beta_n(t)\in \{0,1\}$, the optimal solution of \eqref{dyn_opt} is:
\begin{align} \label{policy_q0}
\beta_n(t)=
\begin{cases}
1,~&q_n(t)E_n(t) \leq \frac{V\gamma(t)}{N},\\
0, ~ &q_n(t)E_n(t) > \frac{V\gamma(t)}{N}.  
\end{cases}
\end{align}

Define the $T$-slot drift as $\Delta_n^{[T]}\triangleq L_n(T)-L_n(1)=\frac{1}{2}q_n^2(T)$. Then the $T$-slot drift-plus-penalty function can be bounded by:
\begin{align} \label{t_slot}
\Delta_n^{[T]}+V\sum_{t=0}^{T-1}u_n(t)&\leq \sum_{t=0}^{T-1} \left(\frac{1}{2} \alpha_n^2+q_n(t)y_n(t)\right) 
+V\sum_{t=0}^{T-1}u_n(t)\nonumber\\
&\leq \frac{T}{2} \alpha_n^2 + \sum_{t=0}^{T-1} \left(q_n^*(t)-q_n(0)\right)y_n^*(t)+VTu_n^* \nonumber\\ 
&\leq  \frac{T}{2} \alpha_n^2 +\sum_{t=0}^{T-1} t\alpha_n^2+VTu_n^*= \frac{T^2}{2}\alpha_n^2 +VTu_n^*,
\end{align}
where $u_n^*$, $q_n^*(t)$ and $y_n^*(t)$ are the optimal utility of $\mathcal{P}3$ obtained by the optimal genie-aided policy, and the corresponding value of queue and energy deficit. The inequality is obtained since the proposed algorithm minimizes $q_n(t)y_n(t)+Vu_n(t)$ in each round.

Since $\Delta_n^{[T]}\geq 0$ and $u_n(t)\geq 0$, from \eqref{yn_sum} and \eqref{t_slot}, we have
\begin{align}
	\sum_{t=0}^{T-1}\beta_n(t)E_n(t)-T\bar{E}_n&\leq q_n(T)\leq \sqrt{T^2\alpha_n^2+2VTu_n^*}, \nonumber\\
	u_n^{\dagger}&\leq u_n^*+ \frac{T}{2V} \alpha_n^2 . \label{bound_un}
\end{align}
By summing \eqref{bound_un} over $n=1,\cdots,N$, we prove Theorem 1.

\end{document}